\documentclass[12pt,letterpaper]{amsart}
\usepackage{mathabx}
\usepackage{graphicx}

\usepackage{float}
\usepackage[T1]{fontenc}
\usepackage{mathtools}
\usepackage{epic,eepic,latexsym, amssymb, amscd, amsfonts}
\usepackage[all]{xy}
\usepackage{url}
\usepackage{hyperref}
\usepackage{enumerate}

 \newlength{\baseunit}               
 \newcount{\numlines}                
\makeatletter
\newcommand*{\rom}[1]{\expandafter\@slowromancap\romannumeral #1@}
\makeatother

\numberwithin{equation}{section}
\newtheorem{thm}{Theorem}[section]
\newtheorem{lem}[thm]{Lemma}
\newtheorem{defi}[thm]{Definition}
\newtheorem{prop}[thm]{Proposition}
\newtheorem{cor}[thm]{Corollary}
\newtheorem{rem}[thm]{Remark}
\newtheorem{ex}[thm]{Example}

\def\R{\mathcal R}
\def\A{\mathcal A}
\def\L{\mathcal L}
\def\P{\mathbb P}
\def\N{\mathbb N}
\def\O{\mathcal O}
\def\n{\mathfrak n}
\def\res{\textup{Res}}
\def\d{\textup{det}}
\def\T{\textup{tr}}
\def\D{\textup{D}}
\def\S{\textup{Sing}}
\def\H{\textup{H}}
\def\a{\alpha}
\def\b{\beta}
\def\t{\theta}

\def\N{\mathfrak N}
\def\X{\mathfrak X}

\def\w{\omega}
\def\dr{\textup{detR$\Gamma$}}
\begin{document}
\pagestyle{plain}
\title{ { Moduli spaces of \MakeLowercase{$q$}-connections and gap probabilities}}
\author{Alisa Knizel\\}

\address{Department of Mathematics\\
  Massachusetts Institute of Technology\\
 Cambridge, MA 02139.}
\email{alisik@math.mit.edu.}

\begin{abstract}
We develop a $q$-analogue of methods introduced by Arinkin and Borodin in
\cite{BA1}, \cite{BA2}. Our goal is to show that the one-interval gap
probability for the $q$-Hahn orthogonal polynomial ensemble can be expressed through a solution
of the asymmetric $q$-Painlev\'e \rom{5} equation. The case of the $q$-Hahn ensemble we consider is the most general case of the orthogonal polynomial ensembles that have been studied. Our approach is based on the analysis of $q$-connections on $\P^1$ with a particular singularity structure. It requires a new derivation of a $q$-difference equation of Sakai's hierarchy \cite{Sak} of type $A_{2}^{(1)}.$ We also calculate its Lax pair. Following \cite{BA2}, we introduce the notion of the $\tau$-function of a
$q$-connection and its isomonodromy transformations. We show that the gap probability function of the $q$-Hahn ensemble can be viewed as the $\tau$-function for an associated $q$-connection and its isomonodromy transformations. 
\end{abstract}
\maketitle

\section{Introduction}

The connection between gap probabilities for orthogonal
polynomial ensembles and Painlev\'e equations was established in
the 90's. In continuous settings the gap
probability function was explicitly written in terms of a specific
solution of one of the six Painlev\'e equations for all classical weights.
It was done in \cite{TW} for the Hermite and
Laguerre weight, in \cite{TW} and \cite{HS} for the Jacobi weight, and
in \cite{W}, \cite{BD} for the quasi-Jacobi weight. 

The first
results for the discrete case were obtained in \cite{BB}.
In present paper we consider the $q$-Hahn orthogonal polynomial
ensemble. We present an explit expression of the gap
probability function for this ensemble in terms of a solution
of the $q$-asymmetric Painlev\'e \rom{5} equation. The $q$-Hahn ensemble is the
most general case of orthogonal polynomial ensembles associated to a
a family of polynomials in the Askey scheme that have been
studied in this context. We believe that the expressions for the gap
probabilities in terms of solutions of Painlev\'e equations in all
known cases should follow from our result through a degeneration
procedure. In the end of the paper we present a numerical evidence
supporting this claim. Considering a specific limit regime we show
that the distribution of the rightmost particle in the $q$-Hahn
ensemble numerically converges to the Tracy-Widom distribution.

We initially became interested in the problem of studying the $q$-Hahn ensemble
as it naturally comes up in the statistical description
of the tilings of
a hexagon by rhombi (see \cite{BGR}) as well as in the representation theory
(see \cite{KR}). Let us describe the former connection in more details.  
\subsection{Tilings of a hexagon and q-Hahn ensemble}
We start by recalling the definition \cite{KS} of $q$-Hahn polynomials.
\begin{defi}
Let $q\in(0,1)$ and $N\in \mathbb Z_{>0}$. Let $0<\a<q^{-1}$ and
$0<\b<q^{-1}$ or $\a>q^{-N}$ and $\b>q^{-N}.$ Define a weight function on
$X= \{0,\dots, N\}$ as
$$\omega(x)=(\a\b q)^{-x}\frac{(\a q,q^{-N};q)_x}{(q,\b^{-1}q^{-N};q )_x},$$ where $(y_1,
\dots, y_i;q)_k=(y_1;q)_k\cdots (y_i;q)_k$, and
$(y;q)_k=(1-y)(1-yq)\cdots (1-yq^{k-1})$ is the
$q$-Pochhammer symbol.
\end{defi}
\begin{defi} Let us fix an integer $k > 0.$ The $q$-Hahn orthogonal ensemble is a probability measure on the set of all $k$-subsets of
$\mathfrak X= \{y_i=q^{-i}: i=0,\dots, N\}$ given by 
$$\mathcal P(y_{x_1},\dots ,y_{x_k})=
\frac{1}{Z} \prod_{1\leq i<j \leq k} {(y_{x_i}-y_{x_j})^2}\cdot \prod_{i=1}^{k}{\omega(x_i)},$$
where $Z$ is the normalizing constant.
\end{defi}

The collection $(y_{x_1},\dots ,y_{x_k})$ is often referred to as
$k$-particle configuration.
\begin{figure}[h]
\includegraphics[width=0.33\linewidth]{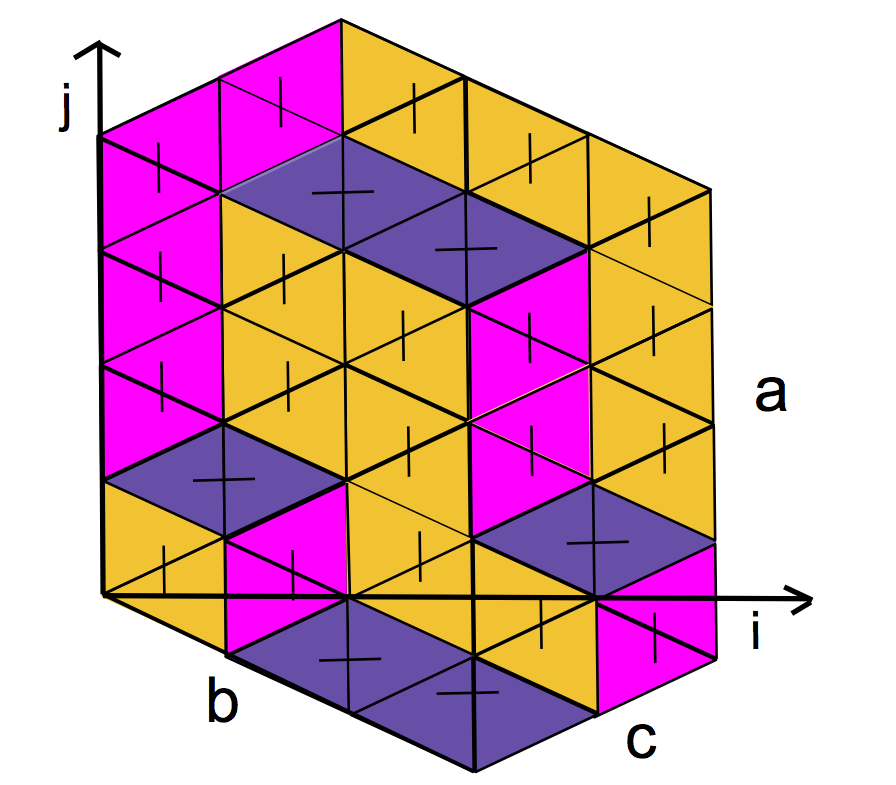}
\caption{Tiling of a $4\times 3\times 2$ hexagon.}
  \label{fig:tile}
\end{figure}

The $q$-Hahn orthogonal ensemble is closely related to the following
tiling model. For any integers $a,$ $b,$ $c\geq 1$ consider a hexagon
drawn on a regular triangular lattice (see Figure \ref{fig:tile}). The tilings
of the hexagon by rhombi are obtained by gluing two neighboring
elementary triangles together (such rhombi are called lozenges). In
Figure \ref{fig:tile} one can also view rhombi as faces of
unit cubes and see the three-dimensional shape
corresponding to a tiling. It is called a 3-D Young diagram or,
equivalently, boxed plane partition.

Consider the set of all tilings
of a hexagon by rhombi, which we denote by $\Omega(a,b,c).$ Let us
denote the purple lozenges by $\Diamond$ and  introduce coordinate
axes $(i, j)$ shown in
Figure \ref{fig:tile}. Define a
probability measure on $\Omega(a,b,c)$
as
\begin{equation}
 P(\mathcal T \in \Omega(a,b,c))=\frac{w(\mathcal
  T)} {\sum\limits_{ \mathcal S \in \mathfrak T}^{ } w(\mathcal
  S)},\text{  where  } w(\mathcal T)=\prod\limits_{\Diamond \in \mathcal
  T}^{ } w(\Diamond),
\end{equation}
and $w(\Diamond)=q^{-j_{\text{top}}}$
 $($ $j_{\text{top}}$ is the $j$-coordinate of the topmost point of the purple lozenge$).$

In the language of boxed plane partitions, this definition assigns to a
plane partition of volume V (=number of $1\times 1\times 1$ boxes) the
probability proportional to $q^{-V}.$ If we send $q\rightarrow 1$ we
obtain a uniform distribution on the set of all tilings.
 
Consider a transformation of the hexagon by means of the following
affine transformation:

 \begin{figure}[h]
\includegraphics[width=0.65\linewidth]{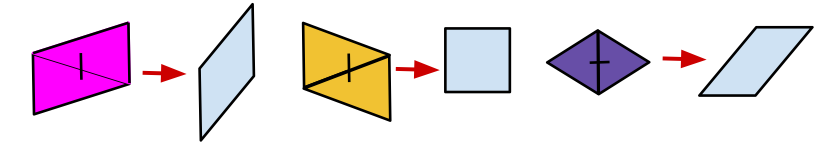}
  \label{fig:aff}
\end{figure}  

In this way it is easy to see the bijection between the tilings and
the set of non-intersecting paths on the plane
lattice \cite{J}. Figure $\ref{fig:hex}$ illustrates this.
 \begin{figure}[h]
\includegraphics[width=0.60\linewidth]{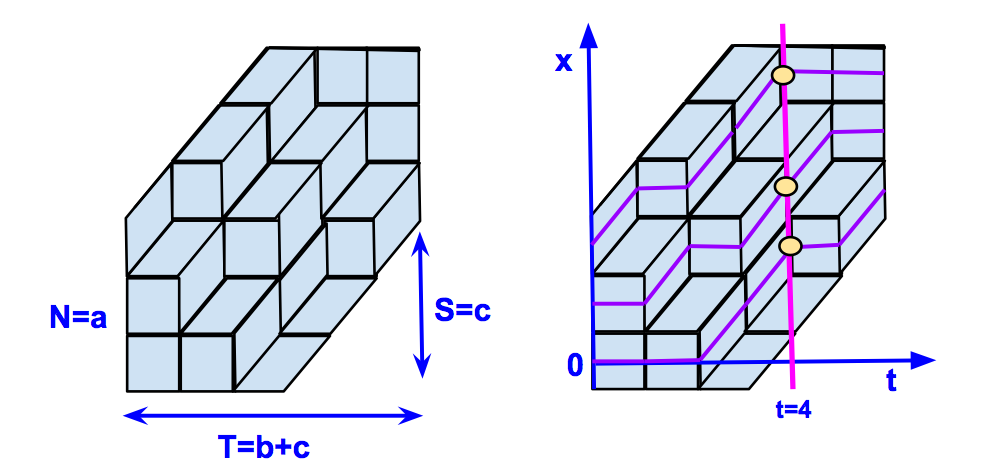}
\caption{Modified tiling of a $3\times 3\times 3$ hexagon and the
corresponding family of non-intersecting paths.}
  \label{fig:hex}
\end{figure}  

Let us introduce new variables $T=a+b,$ $N=a,$ $S=c$ and
coordinates $(t, x)$ (see Figure $\ref{fig:hex}$). Let the
parameters $a,$ $b,$ $c$ be fixed. Consider the
 section of the modified hexagon for some $t$. Then the coordinates of
  the nodes $($intersection of the paths with the section$)$ belong to $$\mathcal F_t=\{x\in \mathbb Z\colon max(0,
    t+S-T)\leq x \leq min(t+N-1, S+N-1)\}.$$ Note that there are
    exactly $N$ nodes for any section $t$.
Denote the configuration of the nodes by $C(t).$

\begin{thm} [\cite{BGR}, Theorem 4.1]
$$\textup{Prob}\{C(t)=(x_1,\dots, x_N)\}=const\cdot\prod\limits_{i<
  j}^{ }(q^{-x_i}-q^{-x_j})^2\prod\limits_{i=1}^{N}w_{t, N, S, T}(x_i),$$
where
$w_{t, N, S,T}(x)$ is the weight function of
the $q$-Hahn orthogonal polynomial ensemble.
\end{thm}

Thus, the study of the $q$-Hahn ensemble gives a lot of insight into
the tiling model. Let us get back to the ensemble itself. We are interested in the following function
$$D_k(s)=\sum_{x_1,\dots, x_k < s}\mathcal P(y_{x_1}, \dots
y_{x_k}),$$
which is called the one-interval gap probability function. We will
call the gap probabilities its values for different $s$. Note that
$D_k(s)=\textup{Prob}(\max\{x_i\}<s)$ and $y_{\max\{x_i\}}$ is called
the position of the rightmost particle.
For the corresponding tiling model using the gap probability function
one can describe positions of the
upmost and bottommost nodes. Therefore, in the limit regime it provides a
lot of information about the frozen boundary, a curve separating
the so-called liquid region from the facets.

 Now we are ready to present our main result.
\begin{thm}\label{pr}
The one-interval gap probabilities $D_k(s)$ satisfy the following recurrence
\begin{equation}\label{qH}
\frac{D_k(s)
  D_k(s-2)}{D_k(s-1)^2}=\frac{\a\b(r_sw-q^{-2s+k+1})(r_{s+1}w-\a^{-1}\b^{-1}q^{-2s-k+1})(t_s-q^{-s+1})^2}{q^{-2s}(q^{-s+1}-q^{-N})(q^{-s+1}-\a
  q)(q^{-s+1}-b^{-1}q^{-N})(q^{-s+1}-q)},
\end{equation}
where $(r_s, t_s)$ is the solution of the asymmetric $q$-Painlev\'e
\rom{5} equation 
\begin{align*} (r_{s+1}t_s +1)(r_st_s+1) &=\frac{q^{-2s}(t_s-q^{-N})(t_s-b^{-1}q^{-N})(t_s-aq)
  (t_s-q)}{\a\b^{-1}q^{-2N}(t_s-q^{-s+1})^2},\\
(r_{s+1}t_{s+1}+1)(r_{s+1}t_s+1) &=\\
q^{-4s+1}& \frac{(q^{-N}r_{s+1}+1)(\b^{-1}q^{-N}r_{s+1}+1)(
 \a qr_{s+1}+1)(qr_{s+1}+1)}{\a\b(r_{s+1}w-q^{-2s+k}) (r_{s+1}w-\a^{-1}\b^{-1}q^{-2s-k+1})},
\end{align*}
and the initial conditions can be found explicitly using Proposition
$\ref{initial}$ and Proposition
$\ref{D}$ below.
\end{thm}

This result predicts the appearance of the Painlev\'e transcendents in the 
limit regime. In Section \ref{sec:7} we consider a limit regime which
corresponds to a simultaneous linear growth of the sides of the hexagon
and present a numerical evidence of the
appearance of the Tracy-Widom distribution near the frozen boundary. We use a method suggested by
Olshanski in \cite{O}.

\subsection{Moduli spaces of $q$-connections} The proof of Theorem
\ref{pr} is based on the ideas introduced in
\cite{BA2}. We start by investigating a notion of the
$\tau$-function of a vector bundle on $\P^1$ equipped with a
$q$-connection. The $\tau$-function we consider is a discrete analogue
of the continuous $\tau$-function introduced by Jimbo, Miwa, and
Ueno in \cite{JMU1}, which proved to play a central role in the theory of the
isomonodromy deformations. In the continuous settings the $\tau$-function is a holomorphic
function on the universal covering space of the space of parameters of the connection,
which vanishes when the corresponding isomonodromy deformation fails
to exist. The $q$-difference Painlev\'e equations
were first studied from the point of view of $q$-isomonodromy
deformations in \cite{JS}.

\begin{defi} Let $q\neq 0,1$ be a complex number. A $q$-connection $\A$ is a pair $(\L, \A(z)), $ where $\L$ is a vector bundle on $\P^1$ and $\A(z)$ is a linear operator $$\mathcal {A}(z)\colon\mathcal{L}_z \rightarrow \mathcal{L}_{qz}$$
that depends on a point $z \in \mathbb{P}^1\setminus \{\infty\}$ in a
rational way $($in particular, $\A(z)$  is defined for all $z\in \mathbb
C$ outside of a finite set$)$; here $\mathcal{L}_z$ is the fiber of $\mathcal{L}$ over $z\in \mathbb {P}^1$. In other words, $\mathcal{A}(z)$ is a rational map between the vector bundle $\mathcal{L}$ and its pullback via the automorphism $\mathbb{P}^1 \rightarrow \mathbb{P}^1$ that sends $z\rightarrow qz$.
\end{defi}
In speaking of $q$-connections, we will frequently abuse the notation
and write $\A(z)$ instead of $(\L,\A(z))$ when there is no danger of confusion.

\begin{defi}
 We say that a point $z_0 \in \P^1$ is a pole of $\A$ if $\A(z)$ is not regular at $z=z_0$. We say that $z_0\in \P^1$ is a zero of $\A$ if the map
$\A^{-1}(z) : \L_{qz} \rightarrow \L_z$
is not regular at $z = z_0$. Note that $\A$ can have a zero and a pole at the same point.
\end{defi}

We proceed by associating to the $q$-Hahn ensemble a vector bundle on $\P^1$ equipped
with a $q$-connection and constructing a sequence of its modifications.

\begin{defi}\label{mod}
Suppose $\mathcal R\colon \L\rightarrow\hat\L$ is a rational
isomorphism between two vector bundles on $\P^1$. We say that $\hat\L$
is a modification of $\L$ on a finite set $S\subset \P^1$ if $\mathcal
R(z)$ and $\mathcal R^{-1}(z)$ are regular outside $S$. We call $\hat\L$ an
upper modification of $\L$ if $\mathcal R$ is regular $($resp. $\L$ is
the lower modification of $\hat \L$$)$.  

 A $q$-connection $(\L,\A(z))$ induces a $q$-connection $(\hat\L, \hat\A(z)),$ which we also call a modification.
\end{defi}
The next step is computation of ``the second logarithmic derivative'' of the
$\tau$-function for the sequence of modifications and proving that it coincides
with ``the second logarithmic derivative'' of the gap probability function. 

In order to follow our plan, we study a special class of $q$-connections in the
spirit of \cite{BA1}. More precisely, we construct moduli spaces of
$q$-connections with a particular singularity structure. The choice of
the singularity structure we consider is dictated by the singularity
structure of the $q$-connection associated to the $q$-Hahn
ensemble.

\begin{rem}
After we choose a trivialization of $\L$ restricted to $\mathbb
A^1=\P-\{\infty\},$ the operator $\A(z)$ can be written in coordinates as a
matrix-valued function $A(z).$ We will call it a
matrix of the $q$-connection. For two different trivializations the
corresponding matrices differ by a $q$-gauge transformation
\begin{equation}\label{gauge1}
\hat A(z)=R(qz)A(z)R^{-1}(z).
\end{equation}
Thus, the moduli space of $q$-connections can be identified with the
moduli space of their matrices modulo $q$-gauge transformations.
\end{rem}

Let us consider $q$-connections $(\A(z), \O\oplus\O(-1)).$ We
impose a number of restrictions on them. After choosing a trivialization the matrices of the
$q$-connections are $2\times 2$ matrices with polynomial entries of
the following form:
$$ A(z)=\left [\begin{array}{cc} a_{11} & a_{12}\\ a_{21} & a_{22}
\end{array}\right],\quad  A(0)=\left [\begin{array}{cc} w & 0\\ 0 & w
\end{array}\right ],$$
$$deg(a_{11})\leq 3,\quad deg(a_{12})\leq 4,\quad deg(a_{21})\leq 2,\quad deg(a_{22})\leq 3 $$
$$\d (A(z))=uv(z-a_1)(z-a_2)(z-a_3)(z-a_4)(z-a_5)(z-a_6),$$
where $a_1,\dots a_6, u,v,
w$ are complex parameters.
Denote $$ S(z)=\left [\begin{array}{cc} 1 & 0\\ 0 & \frac{1}{z}\end{array}\right ].$$
We also require that 
$$\d (S(qz)^{-1}A(z)S(z))=quvz^6+\mathcal{O}(z^5)\text{  and  } \T
(S(qz)^{-1}A(z)S(z))=(u+qv)z^3+\mathcal O(z^2).$$
Note that $S$ is essentially a basis of $\O\oplus\O(-1)$ in the
neighborhood of $\{\infty\} \in \P^1$. We say that such $q$-connections are of type $\lambda=(a_1,\dots, a_6; u,qv;
w, w;3).$

We consider $q$-connections modulo $q$-gauge transformations of the form
$\eqref{gauge1}.$ We can write $R$ in coordinates as 
$$R(z)=\left [\begin{array}{cc} r_{11} & r_{12}\\ 0 & r_{22}
\end{array}\right],$$
$$deg(r_{11})\leq 1,\quad deg(r_{12})\leq 2,\quad deg(r_{22})\leq1. $$

\begin{lem}
\label{mo}
Under certain nondegeneracy conditions on the parameters $a_1,\dots,
a_6, u,v, w$ of a $q$-connection $\A,$ there exits its unique
modification $\hat{\A}$ of type $(qa_1, qa_2\dots, a_6; u,qv;
qw, qw;3).$
\end{lem}

Let us assume that the parameters $a_1,\dots,
a_6, u,v, w$ are generic. We will discuss what it means for the parameters to be generic
in Section \ref{sec:3}. We proceed by showing that the moduli space $M_{\lambda}$ of $q$-connections of type $\lambda=(a_1,\dots
a_6; u,qv; w,w;3)$ modulo $q$-gauge transformations is two-dimensional and its smallest smooth
compactification can be identified with $\P^2$ blown-up at nine
points; more precisely, it is a Sakai surface of type $A_{1}^{(2)}.$ 
Let us state our next result.
\begin{thm}[$q$-$P_V$]\label{qPV}
Consider the second root of the matrix element $a_{21}$ $($note that $0$ is
always a root of $a_{21}$$)$ as the first
coordinate on $M_{\lambda}$, denote it by $t$. Let the second coordinate
be $$r=\frac{w(t-a_3)(t-a_4)(t-a_5)(t-a_6)}{a_{11}(t)a_3a_4a_5a_6t}-\frac{1}{t}.$$
Consider the modification of $\L$ to
$\hat\L$ from Lemma \ref{mo}. It shifts 
$$a_1\rightarrow qa_1, \quad a_2\rightarrow qa_2,\quad w\rightarrow qw.
$$
Then this modification defines a regular morphism $\textup{qPV}$ between two moduli spaces
$M_{\lambda}$ and
$M_{\hat\lambda},$ the moduli space of $q$-connections of type $$\hat{\lambda}=(qa_1, qa_2,\dots, a_6; u,qv;
qw, qw;3).$$
Moreover, the coordinates $(\hat t,\hat r)$ on the moduli space
$M_{\hat\lambda}$ are related
to $(t,r)$ by the asymmetric $q$-Painlev\'e V equation
\begin{align} \label{recc1}
(r\hat t+1)(rt+1)&=\frac{uva_1^2a_2^2(ra_3+1)(ra_4+1)(ra_5+1)(ra_6+1)}{(rw-va_1a_2)(qrw-ua_1a_2)},\\
(\hat r \hat t +1)(r\hat t+1)&= \frac{a_1a_2(\hat
t-a_3)(\hat t-a_4)(\hat t-a_5) (\hat t-a_6)}{a_3a_4a_5a_6(q\hat
t-a_1)(q\hat t-a_2)}. \nonumber
\end{align}
\end{thm}
In the proof we also explicitly compute a Lax pair for this equation in a new
fashion.

The asymmetric $q$-Painlev\'e \rom{5} equation was first found by
B. Grammaticos, A. Ramani
and their co-authors \cite{GR}. The form of the equation \ref{recc1} we use
differs from the standard one by a change of notation. In \cite{Sak2} Sakai showed that the asymmetric $q$-Painlev\'e \rom{5} equation appears
as a particular case of the four-dimensional $q$-Garnier system
($N=2$) and first expressed it in a Lax formalism. The computations
done by Sakai are not applicable in our situation since the
modifications of the parameters we aim to study are different from the ones he
considered, see Section \ref{sec:4} for details. The other two
approaches to obtaining Lax pairs were presented in \cite{OW},
\cite{Y}.

We believe that $q$-orthogonal ensembles, for which the similar
results were obtained in \cite{BB}, also fit in
the framework of $q$-connections.

The paper is organized as follows. 

In Section \ref{sec:3} we discuss the
properties of $q$-connections and show that the geometric approach to
isomonodromy deformation of $q$-connections implies that the
transformations lift to isomorphisms between surfaces. We also define
modifications of
a special class of $q$-connections which will lead to the
asymmetric $q$-difference Painlev\'e $V.$ 

Then in Section \ref{sec:4} we give the proof of this result. 

Section \ref{sec:5} is devoted to the discussion of the $q$-analogue of the
constructions introduced in \cite{BA2}. The $q$-generalization is
mostly straightforward, however, the computations require some
work. We conclude this section by providing the recipe for computing the $\tau$-function of
the isomonodromy transformations leading to the
asymmetric $q$-Painlev\'e \rom{5} equation.

In Section \ref{sec:6} we investigate the $q$-Hahn orthogonal polynomial ensemble
and give a proof of Theorem \ref{pr}.

In Section \ref{sec:7} we present a numerical evidence of the
appearance of the Tracy-Widom distribution near the frozen boundary of
the tiling model of a hexagon by rhombi with the
weight proportional to $q^{-Volume}$.

\subsection*{Acknowledgements}
The author is very grateful to Alexei Borodin for suggesting the
project and for many helpful discussions.

\section{Moduli spaces of $q$-connections}
\label{sec:3}
\subsection{Preliminaries.}
\begin{defi}
Let $\L$ be a rank $m$ vector bundle on $\P^1$ and $\A(z)$ be a $q$-connection on $\L$. Suppose $\A(z)$ satisfies the following conditions:
\begin{enumerate}[(i)]
\item The only zeroes and poles of $\A(z)$ are as follows: a pole of
  order $n$ at infinity and simple zeroes at $k$ distinct points
  $a_1,\dots, a_k\in \mathbb A^1$ $($i.e., $\A(z)$ is regular and $\d
  (\A(z)) $ has zero of order one at $a_i$$)$.
\item On the formal neighborhood of $\infty\in\P^1$, there exists a trivialization $\R(z)\colon \mathbb C^2\rightarrow \L$ such that the matrix $\A$ with respect to $\R$ satisfies 

$$\R(qz)\A(z)\R^{-1}(z)=A_nz^n+\textup{O}(z^{n-1}),$$
where $\A_n$ is a semi-simple matrix with eigenvalues $\rho_1,\dots, \rho_m$.
\item On the neighborhood of zero $\A(z)|_{z=0}=A_0$ is semi-simple with eigenvalues $\t_1,\dots,\t_m.$
\end{enumerate}
We call  $q$-connection $\A(z)$ or, more precisely, the pair $(\L,\A(z))$ a $q$-connection of type $\lambda =(a_1, \dots, a_k; \rho_1, \dots, \rho_m; \theta_1, \dots, \theta_m;n).$
\end{defi}

\begin{rem}
We will also consider $q$-connections that have simple poles besides
simple zeroes. The operation ``multiplication by scalar'' turns a pole
into a zero and vice versa. More precisely, let $f(z)\neq 0$ be a
rational function on $\P^1$, and let $\A(z)$ be a $q$-connection on
$\L$. Then the product $f(z)\A(z)$ is again a $q$-connection on $\L.$
For any $q$-connection $\A(z)$, we can choose $f(z)$ such that the only pole of the product $f(z)\A(z)$ is at infinity. 
In this case we will slightly modify the notation for the type of the connection. We will write $\A$ is of type $(a_1, \dots, a_k; b_1,\dots, b_l;\rho_1, \dots, \rho_m; \theta_1, \dots, \theta_m;n)$, where $(a_1, \dots, a_k)$ and $(b_1,\dots, b_l)$ is the collection of zeroes and poles respectively.

\end{rem}

\begin{defi}
Let $\A(z)$ be a $q$-connection on $\L$ of type $(a_1, \dots,
a_k;\rho_1, \dots, \rho_m; \theta_1, \dots, \theta_m;n).$  The
$q$-degree of $\A(z)$ is the following quantity 
$$\textup{deg}_q(\A(z))=\frac{\prod^{k}_{i=1} a_i \prod_{i=1}^{m} \rho_i}{
  \prod^{m}_{i=1}\t_i}.$$
\end{defi}

\begin{lem}\label{irr}
Let $\A(z)$ be a $q$-connection on a rank $m$ vector bundle $\L$
of type $$(a_1, \dots, a_k;\rho_1, \dots, \rho_l; \theta_1, \dots , \theta_m;n)$$ with invertible $A_0$ and $A_n.$ Then the following holds
$$mn=k-l\quad \text{and} \quad q^{-\textup{deg}(\L)}=\textup{deg}_q(\A(z)).$$
\end{lem}
\begin{proof}
 Determinant of the morphism $\L_{z}\rightarrow \L_{qz}$ is a map
 $\d(\A(z)):\bigwedge ^m\L_z\rightarrow \bigwedge^m \L_{qz}$. In this
 way, we can define a $q$-connection $\d(\A)$ on a line bundle
 $\d(\L).$ Note that $\d(\A)$ is of type $(a_1, \dots, a_k;\prod_{i=1}^{m}\rho_i; \prod_{i=1}^{l}\t_i;k-l).$
\end{proof}

\begin{cor}\label{cirr}
Suppose $\A(z)$ is a $q$-connection on rank $2$ vector bundle $\L$ of
type $$(a_1, \dots, a_k; b_1,\dots, b_l;\rho_1, \rho_2; \theta,
\theta;n)$$ with invertible $A_0$ and $A_n.$ Suppose that for any
$I\subset\{1,\dots, k\}$ we have $\frac{\rho_j }{\t}\prod\limits_{i\in
  I}a_i\neq q^k$ for any $k\in \mathbb Z$ and $j=1,2$.  Then $\A(z)$
is irreducible, i.e., there is no rank one subbundle $\ell\subset \L$ such that $\A(\ell_z)\subset \ell_{qz}$ for all $z$. 
\end{cor}
\begin{proof}
Suppose that $\ell\in \L$ is an invariant subbundle of rank $1$. Then
$\A$ induces a $q$-connection $\A_{\ell}$ on $\ell$. All its zeroes
belong to $\{a_1,\dots, a_k\}$ and its type is either $(a_1,\dots,
a_k;\rho_1;\t;n)$ or $(a_1,\dots, a_k;\rho_2;\t;n)$. Now Lemma \ref{irr} leads to a contradiction.
\end{proof}

\begin{lem}
Suppose that $\A$ is an irreducible $q$-connection on a rank $2$
vector bundle $\L$ of the type  $\lambda =
(a_1,...,a_k;\rho_1,\rho_2;\t_1,\t_2;n)$ with invertible $A_0$ and $A_n.$ If $\L\simeq\O(n_1)\oplus\O(n_2),$ then $|n_1-n_2|\leq n.$ Moreover, if $q\rho_1=\rho_2$ the inequality is strict.
\end{lem}
\begin{proof}
Without loss of generality we can assume that $n_1\geq n_2.$ Let
$\ell\subset \L$ be a rank $1$ subbundle of degree $n_1$. Since
$(\L,\A)$ is irreducible, $\ell$ is not invariant, thus, a rational
map $\psi\colon \ell\rightarrow \L\rightarrow s^{*}\L\rightarrow
s^*{\L /\ell}$ is not identically zero (recall that $s\colon z\in
\P^1\rightarrow qz\in \P^1$). Notice that $\psi$ can have at most an
order $n$ pole at $\infty$ (and no other poles). Therefore,
$n_1=\textup{deg}(\ell)\leq n+\textup{deg}(\L/\ell)=n+n_2.$ The second statement follows
from the fact that if $q\rho_1=\rho_2$ then $\psi$ has at most an order
$n-1$ pole, because the coefficient of $z^n$ in $\psi$ is an
off-diagonal element of a scalar matrix.
\end{proof}

From now on we will work with $q$-connections on rank $2$ vector bundle.
Denote by $M_{\lambda}$ the moduli space of $q$-connections of type
$\lambda=(a_1, \dots, a_k; b_1,\dots, b_l;\rho_1,\rho_2; \t, \t;2).$  From Lemma $\ref{irr}$ we see that $M_{\lambda}$ is empty unless 

\begin{equation}\label{f}
k=2n,
\end{equation}
\begin{equation}
\textup{deg}_q(\lambda)=q^k \text{, where } k \text{ is an integer }.
\end{equation}
Let us also consider the following assumptions on $\lambda$:
\begin{equation}\label{ir}
\frac{\rho_j\prod_{i\in I}a_i}{\t}\neq q^k \text{ for any }
I\subset\{1,\dots, k\},\text{  any } k\in \mathbb Z \text{ and }
j=1,2;
\end{equation}
\begin{equation}\label{l}
\rho_1,\rho_2\neq 0 \quad \theta_1=\theta_2\neq 0.
\end{equation}
Assumption ($\ref{ir}$) implies irreducibility of the $q$-connection
and can be used to prove that the moduli space $M_\lambda$ is a smooth
variety. In the present paper we do not work with the categorical
definition of the moduli space since it is not necessary for our purposes. 
However, we think that one can prove that the moduli space we consider
is the coarse moduli space following the arguments in \cite{AL}.

In what follows we will need more facts about the $q$-connections.

\begin{defi}
Let $\Lambda^{r}$ be the set of all q-connections satisfying \ref{f}--\ref{l}.
\end{defi}

\begin{defi}
Let $\L$ be a rank two vector bundle on $\P^1$, $x\in \P^1$ and $\ell\subset
\L_{x}$ a one-dimensional subspace. Denote by $L$ the sheaf of
sections of $\L$. An elementary lower modification of $\L$ at $x$
with respect to $\ell$ is a rank two
bundle $\hat \L$ whose sheaf of sections is $\hat L=\{s\in
L| s(x)\in \ell\}.$

We can dually define elementary upper modifications.
\end{defi}
Let us denote the set of poles and zeros of $\A$ by $\S(\A).$
\begin{lem} $($Lemma 1.7 \cite{BA2}$)$ Suppose that $q^{-1}x$ is not a
  singular point of $\A,$ while $x$ is a singular point of $\A$. Then there exists a unique modification
  $\A^{\{x\}}$ such that 
\begin{enumerate}[(i)]
\item $\S(\A^{\{x\}}) = (\S(A) \ {x}) \cup {q^{-1}x};$
\item $\A$ is the unique modification of $\A^{\{x\}}$ at $x$ with no
  singularity at $q^{-1}x$. 
\end{enumerate}
The dual statement also holds.
\end{lem}

\begin{lem}
Suppose $(\L,\A)\in M_\lambda$ for $\lambda =\lambda=(a_1,...,a_k;\rho_1,\rho_2;\t,\t;n)\in \Lambda^r$. Let $(\hat\L, \hat\A)$ be an elementary upper modification of $\L$ at $x\in \P^1$. Then the only cases when $(\hat\L, \hat\A)$ belongs to $M_{\hat\lambda}$ for some $\hat\lambda\in\Lambda^r$ are as follows:
\begin{enumerate}[(i)]
\item If $x=\infty$, then $\hat \lambda = (a_1,\cdots ,ak;q^{-1}\rho_1,\rho_2;\t,\t;n)$ or $\hat \lambda = (a_1,...,a_k;\rho_1,q^{-1}\rho_2;\t,\t;n).$
\item If $x=a_i$ is a zero of $\A$ and $q^{-1}x\neq a_j$ is not,\\ then $\lambda = (a_1,\cdots,q^{-1}a_i,\cdots,a_k;\rho_1,q\rho_1,\t,\t;n).$
In either case, the elementary modifications define an isomorphism $M_{\lambda}\rightarrow M_{\hat\lambda}$.
\end{enumerate}
\end{lem}

\subsection{Geometric description of $M_{\lambda}$}

We will focus on $q$-connections of type
$$\lambda=(a_1,...,a_k;u,qv;w,w;3)\in \Lambda^{r}$$ on $ \L\simeq \O\oplus
\O(-1).$ This choice is dictated by appearance of this type of
connections in our analysis of the $q$-Hahn ensemble. 

Let us choose an isomorphism $\mathcal S\colon\L\simeq \O\oplus \O(-1).$ Then $\A$ induces a $q$-connection of type $\lambda$ on $\O\oplus \O(-1).$ Such a $q$-connection is represented by a matrix of the following form 
\begin{equation}\label{A}
A=\left [\begin{array}{cc} a_{11} & a_{12}\\ a_{21} & a_{22}
\end{array}\right],
\quad a_{11}, a_{22}\in\Gamma(\P^1, \O(3)), \quad a_{12}\in \Gamma(\P^1,
\O(4)), \quad a_{21}\in \Gamma(\P^1, \O(2)). \nonumber
\end{equation}

This representation is not unique: $\mathcal S$ can be composed with an automorphism of the bundle. Such an automorphism can be written as a matrix

\begin{equation}\label{gauge}
R=\left [\begin{array}{cc} r_{11} & r_{12}\\ 0 & r_{22}
\end{array}\right],\\
\quad r_{11}, r_{22}\in \mathbb C-\{0\}, \quad r_{12}\in \Gamma(\P^1, \O(1)).
\end{equation}

If we choose another trivialization $\mathcal S$ then $A$ is replaced
by the $q$-gauge transformation 

$$R(qz)A(z)R(z)^{-1}.$$

\begin{lem}\label{cl}
Let $\A$ be a $q$-connection on $\O\oplus \O(-1)$ and let its matrix
be of the form $\eqref{A}.$ Denote $$ S(z)=\left [\begin{array}{cc} 1 & 0\\ 0 & \frac{1}{z}\end{array}\right ].$$
We claim that $\A$ is of type $\lambda=(a_1,...,a_k;u,qv;w,w;n)$ if and only if $A$ satisfies the following conditions
\begin{equation}
\d A(z)=uv(z-a_1)(z-a_2)(z-a_3)(z-a_4)(z-a_5)(z-a_6), \nonumber
\end{equation}
\begin{equation}
\d (S(qz)^{-1}A(z)S(z))=quvz^6+\mathcal{O}(z^5)\text{  and  } \T
(S(qz)^{-1}A(z)S(z))=(u+qv)z^3+\mathcal O(z^2), \nonumber
\end{equation}
$$A(0) \text{ has an eigenvalue } w \text{ of degree } 2.$$
\end{lem}

We can now think of $M_{\lambda}$ as a quotient of all matrices that
satisfy Lemma $\ref{cl}$ modulo $q$-gauge transformations $\eqref{gauge}.$ 

\begin{rem}
In \cite{Sak2} Sakai considered the matrices of the form
$$A(z)=A_0+A_1z+A_2z^2+A_3z^3,$$
where
$$A_3=\left [\begin{array}{cc}\kappa_1 & 0\\ 0 & \kappa_2
\end{array}\right],\quad \d A(z)=\kappa_1\kappa_2(z-a_1)(z-a_2)(z-a_3)(z-a_4)(z-a_5)(z-a_6).$$
He also imposed the condition $q\kappa_1=\kappa_2.$ Sakai considered the
modifications that shift
$$a_1\rightarrow qa_1,\text{ }a_2\rightarrow qa_2\text{ and }
\theta_1\rightarrow q\theta_1, \text{ }\theta_2\rightarrow
q\theta_2, $$
where $\theta_1,$ $\theta_2$ are the eigenvalues of $A_0.$
\end{rem}

\begin{thm}\label{Sak}
The smallest compactification of $M_{\lambda}$ is isomorphic to the
surface of type
$A^{(1)}_{2}$ in Sakai's classification.
\end{thm}

\begin{proof}
Denote by
$t\in \P^1$ the second root of $a_{21}$ (note that zero is always a
root of $a_{21}$). From the irreducibility assumption it follows that
this coordinate is correctly defined. The direct computation shows
that $t$ is invariant under $q$-gauge transformations.

It is convenient to reduce the matrices of a $q$-connection $\A$ of type
$\lambda$ to
some normal form.  Solving a system of linear equations for the
coefficients of $R$ we can reduce $A(z)$ to the following form (we assume $t\in \P^1-\{\infty\}$)

\begin{equation}\label{normform}
\left [\begin{array}{cc} sz+w & h_4z^4+h_3z^3+h_2z^2+h_1z\\ z(z-t) & k_3z^3+k_2z^2+k_1z+w\end{array}\right ].
\end{equation}

Then $s$ can be considered as the second coordinate in this chart. 

Since $\A$ is of type $\lambda,$ we can find all
the coefficients of its matrix in terms of $s$ and $t$. We get 
\begin{enumerate}
\item $ h_1(ts+w)=F_{h_1}(t, s),$
\item $h_2=F_{h_2}(t,s),$
\item $h_3=F_{h_3}(t,s),$
\item $k_1=F_{k_1}(t,s),$
\item $k_2(ts+w)=F_{k_2}(t,s),$
\end{enumerate}
where $F_{-}$ is a polynomial.
Moreover, 
$$F_{h_1}(t,-\frac{w}{t})=-\frac{vw^3\prod^6_{i=1}(t- a_i)}{t^3}\quad
\text{and} \quad  F_{k_2}(t,-\frac{w}{t})=-\frac{vw^2\prod^6_{i=1}(t- a_i)}{t^2}.$$

Thus, in order to resolve the singularities we need to blow-up in six
points $(a_i, -\frac{w}{a_i}).$

Let us work with $t$ in the neighborhood of $\infty.$ Let us denote
$\frac{1}{t}$ by $x.$ In this chart we can reduce the matrix of the
$q$-connection to
another normal form 

\begin{equation}\label{normform}
\left [\begin{array}{cc} yz^3+w & h_4z^4+h_3z^3+h_2z^2+h_1z\\ z(1-xz) & k_3z^3+k_2z^2+k_1z+w\end{array}\right ].
\end{equation}

Let $y$ be the second coordinate.
Again, since $\A$ is of type $\lambda,$ we can find all
the coefficients of its matrix in terms of $x$ and $y$. We get

\begin{enumerate}
\item $ h_1x(x^3w+y)=G_{h_1}(x, y),$
\item $h_2(x^3w+y)=G_{h_2}(x,y),$
\item $h_3(x^3w+y)=G_{h_3}(x,y),$
\item $k_1(x^3w+y)=G_{k_1}(x,y),$
\item $k_2x(x^3w+y)=G_{k_2}(x,y),$
\end{enumerate}
where $G_{-}$ is a polynomial.

Moreover, $$
G_{h_1}(0,y)=wq^{-1}(qv-y)(u-y),\quad
G_{k_2}(0,y)=q^{-1}(qv-y)(u-y).$$
Notice that in the neighborhood $x=0$ $(t=\infty)$ the coordinate $y$ can
not be zero due to our assumptions about the behavior of the
$q$-connection at $\infty.$ 
Therefore, in order to resolve the singularities we need to blow-up in two points $(0,
qv)$ and $(0, u)$.

We covered $M_\lambda$ by two charts $\mathbb
A^{1}\times \mathbb P^{1}-\{\infty\}$ and $\mathbb
A^{1}\times \mathbb P^{1}-\{0\}$ blown-up in a number of points. We can
glue them together along their intersection $q\neq 0$, $\infty.$ The
transition map is $y\rightarrow s=yt^2.$ 

Let us recall that we would like to introduce the following coordinate
on $M_\lambda$
$$r=\frac{w(t-a_3)(t-a_4)(t-a_5)(t-a_6)}{a_{11}(t)a_3a_4a_5a_6t}-\frac{1}{t}.$$
We see that the rational map given by this coordinate
$r\colon M_\lambda \rightarrow \P^1$ is actually regular, since we
have already resolved its singularities.

Summing up, we have realized
$M_\lambda$ as an open subscheme of the blow-up of $\P^1\times \P^1$ in eight points; more
precisely, we can check by direct computation that its compactification is isomorphic to $A^{(1)}_{2}.$
\end{proof}
\section{Proof of Theorem $\ref{qPV}$}
\label{sec:4}
\subsection{q-Difference $P_V$}

In this section we discuss the proof of Theorem $\ref{qPV}$. It
is based on calculations which are hard but straightforward. The
trickiest part is the introduction of the coordinate $r$. 

\begin{proof}
Take $(\L,\A)\in M_{\t}$ and set $(\hat\L,\hat\A)=\textup{qPV} (\L,\A)$. Note that it suffices to check the formulas ($\ref{recc1}$) on a dense subset of $M_{\t}$; we can therefore assume
that $t((\L,\A))\neq 0, \infty$. Then there exists an isomorphism
$\mathcal S\colon\O\oplus\O(-1) \xrightarrow{\sim}\L$ such that
the matrix of $\A$ relative to $\mathcal S$ is of the form 

\begin{equation} \label{matrixform}
A(z)=\left [\begin{array}{cc}s z+w & h_4 z^4+h_3 z^3+h_2z^2+h_1z\\ z(z-t) & k_3z^3+k_2z^2+k_1z+w
\end{array}\right].
\end{equation}

 Let us introduce the coordinate $p$ such
that $s=\frac{p-w}{t}$. An analogue of this coordinate was used in
\cite{JS}, \cite{BA1}, \cite{JM2}, \cite{Sak2}. The specification of the behavior at infinity implies that
$k_3=qu+w$ and $h_4=-uv$. Moreover,
$$\d(A(z))=uv(z-a_1)(z-a_2)(z-a_3)(z-a_4)(z-a_5)(z-a_6),\quad
uv\prod_{i=1}^{6}a_i=w^2.$$ This information allows us to express
$h_3$, $h_2$, $h_1$ and $k_2$, $k_1$ as rational functions in $t$ and
$p$. 

By the definition of $\textup{qPV}$, the matrix $\hat A$ is the $q$-gauge
transformation of $A$:

$$\hat A(z)=R(qz)A(z)R^{-1}(z),$$

where $R$ is of the form 

$$R(z)=\left [\begin{array}{cc}\a_1 z+\a_0 & \b_2 z^2+\b_1 z+\b_0\\
    \gamma_0 & \delta_1 z+\delta_0
\end{array}\right],$$

such that $$\d(R(z))=c (z-a_1)(z-a_2).$$

Also, we know that $\hat A$ has no singularities at $a_1$ and $a_2.$
These restrictions yield polynomial equations on the coefficients of the gauge
matrix. The resulting system determines $R$ up to a multiplicative
constant and we obtain expressions for $\hat t$ and $\hat p$ in terms
of $t$, $p$. The final step is the introduction of the coordinate $r$.

We knew form the geometric considerations (Theorem $\ref{Sak}$) that
we should obtain the asymmetric q-Painlev\'e V. Thus,
we expected the existence of the
coordinates such that the formulas can be written in a nice multiplicative
form. We
investigated properties of the analogue of $r$ in the case of $d$-connections
and then came up with the formula for $r$ in the case of
$q$-connections. First, we were able to express $\hat t$ as a rational function
in $t$ and $p$. The next step was to factor the denominator of this
expression and compute the residues. After some work we obtained the
right candidate for the coordinate $r.$ 

\end{proof}

\begin{rem}
Note that $\hat A(z)=R(qz)A(z)R^{-1}(z)$ is the compatibility
condition for the following system of equations
$$Y(qx)=A(x)Y(x) \text{--- linear } q \text{-difference equation,}$$
$$\hat Y(x)=R(x)Y(x)\text{ --- the deformation equation}.$$
Thus, we have expressed the asymmetric $q$-Painl\'eve V equation in a
Lax pair formalism. 
\end{rem}
\section{Computing the ratios of the $\tau$-function}
\label{sec:5}
\subsection{$\tau$-function}
Following \cite{BA2}, we introduce the notion of the $\tau$-function for a vector bundle.

Let $\d(V)$ denote the top exterior power of a finite dimensional
vector space $V$ and $V^{-1}$ denote its dual. Let us also define a
one-dimensional vector space 

$$\dr(\L)=\d(H^0(\P^1,\L))\otimes ( \d(\H^1(\P^1,\L)))^{-1}.$$

\begin{defi}
Suppose that a rank $m$ vector bundle $\L$ is of degree $-m$. We define $\tau(\L)\in \dr(\L)^{-1}$ by
$$
\tau(\L) = \left\{
        \begin{array}{ll}
            1, & \quad \L \simeq \O(-1)^m; \\
            0, & \quad \text {otherwise}.
        \end{array}
    \right.
$$

\end{defi}

By itself
$\tau(\L)$ contains almost no information, let us show how we can
interpret the ratios of the $\tau$-function in some specific situations.

\begin{ex}\label{iso}
Let us consider a vector bundle $\L$ with slope $-1$ and its upper
modification $\hat\L$ such that $\H^{1}(\P^1,\hat\L)$ vanishes. 
Then the long exact cohomology sequence for 
$$0\rightarrow\L\rightarrow\hat\L\rightarrow \hat\L/ \L\rightarrow 0$$
provides an identification $\dr(\L)=\d(H^0(\P^1,\hat\L)\otimes (
\d(\H^1(\P^1,\hat\L/\L )))^{-1}.$ If we now consider the
morphism $${\n} \colon
\H^0(\P^{1},\hat \L)\rightarrow \H^0(\P^{1},\hat \L/\L),$$
then $\tau(\L)$ can be identified with $\d(\n)$. 
\begin{rem}
Since the slope of $\L$ is $-1,$ we see that $\H^0(\P^{1},\L)=0$ if and
only if $\L=\O(-1)^m$.
\end{rem}
\end{ex}

\begin{ex} Let $\L_2$ be a
modification of $\L_1$. We assume that $\L_1$ is isomorphic
to $\O(-1)^m$ and $\L_2$ is of slope $-1$.  One can always find an
upper modification $\hat\L$ of $\L_1$ that is also an upper
modification of $\L_2$. Then the ratio
$$\frac{\tau(\L_2)}{\tau(\L_1)}\in\dr(\L_1)\otimes\dr(\L_2)^{-1}$$ corresponds
to the determinant of the composition
$$ \H^0(\P^{1},\hat \L/\L_1)\xrightarrow{\sim}  \H^0(\P^{1},\hat \L) \rightarrow \H^0(\P^{1},\hat \L/\L_2).$$
\end{ex}

We will see below that if the vector bundles are equipped with
$q$-connections then the second ratio ("second logarithmic
derivative") of the $\tau$-function makes sense as a number.

\subsection{General $q^{-1}$-connections.}
\begin{rem}
We will be using $q^{-1}$-connections in the rest of the paper to make
all the computations parallel to the ones done in \cite{BA2}. 
\end{rem}
Let $\L=\O(-1)^{m}$ be a vector bundle and $\A$ be a $q^{-1}$-connection on
$\L$. Suppose that $\A$ has singularities at $n$ distinct points
$a_1$, $\dots, a_n$ and no singularities at $q^ka_i$ for $k\in \mathbb
Z-\{0\}$, $i=1,\dots, n$. We impose no restrictions on behavior of
$\A$ elsewhere. 

Consider 

$$\Omega\colon=\{\a_u=(q^{-u_1}a_1, \dots, q^{-u_n}a_n)| u=(u_1, \dots,
  u_n)\in \mathbb Z^{n}-\{0\}\}\subset \mathbb C^n.$$

Let $\a_u\in \Omega.$ We claim that there exists a unique
modification $\L_{u}$ of $\L$ such that 

$$\S(\A_u) =\S(\A)\setminus \{a_1,\dots, a_n\}\cup\{q^{-u_1}a_1,\dots, q^{-u_n}a_n\}. $$
This can be done by a sequence of elementary upper and lower modifications. 
Let us set 

$$S_u^{(i)}\colon=\dr(\L_u)^{-1}\otimes \dr(\L_{u-e_i}).$$

\begin{prop} \label{id} For any $u$, $v\in \Omega$ there exists a canonical
  isomorphism $S^{(i)}_{u}\simeq S^{(i)}_v.$

\begin{proof}
The construction and the proof tautologically repeats
\cite{BA2}(Proposition 2.1).
\end{proof}

\end{prop}
Let us now define the ratios of $\tau$. Set
$$\mathbb Z_0^{n}=\{t=(t_1,\dots, t_n)| \sum_{i=1}^{n}
\kappa_i s_i=0\},$$
where $\kappa_i$ are the degrees of the singularities $a_i$ of $\A.$

For any $a_{u}\in \Omega$ such that $\L_{u}$ is of slope -1 we see
that $\L_{u+s}$ is also of slope $-1$ for $s=(s_1,\dots, s_n)\in \mathbb
Z_0^{n}$. Once $t$ is fixed, ``the second logarithmic derivative'' of the $\tau$-function
$$
\frac{\tau(\L_{u+t+s})}{\tau(\L{u+t})} \cdot
\frac{\tau(\L_{u})}{\tau(\L_{u+s})}  \in \bigotimes
(S^{(i)})^{\otimes s_i} \otimes \bigotimes(
(S^{(i)})^{\otimes s_i})^{-1}
$$
makes sense as a number.
Here we used Proposition \ref{id} to be able to drop the subscript of $S^{(i)}$.

\subsection{Computations in coordinates for connections with simple zeros.}
Now let us make the construction above more explicit. All the
computations are parallel to the ones in \cite{BA2}.

We will present the result for one special case which we will use later.

Let us fix an isomorphism $\L\simeq \O(-1)\oplus\O(-1).$ Then the
$q^{-1}$-connection $\A $ is given by its matrix $A(z).$
Suppose $A$ has singularities at $a_1, \dots, a_n$ and
no singularities at $q^ka_i$ for any $k\in\mathbb Z - \{0\}.$
The singularities at $a_1,\dots, a_N$ are simple zeroes and poles.
Let $a_1$ be a simple zero and $a_2$ a simple pole.

For generic $A(z)$,
there exists a unique rational matrix $R(z)$ with the following
properties

\begin{itemize}
\item $\hat \A(z)=R(q^{-1}z)A(z)R^{-1}(z)$ has the same singularity
  structure as $A(z)$ with singularities at $\hat a_1=qa_1, \hat
  a_2=qa_2, \hat a_3=a_3, \dots, \hat a_n=a_n;$
\item $R(\infty)=I;$
\item Both $R(z)$ and $R^{-1}(z)$ have singularities only at $a_1$ and $a_2.$
\end{itemize}

It is not hard to construct such a matrix $R(z).$ Let us choose a
basis $w$ in the kernel $A(a_1)$ and a basis $w\textprime$ in the image of
$\res_{a_2} A^{T}(z)$, by $A^T$ we denote the matrix transpose. Then 

$$R(z)=I+\frac{R_0}{z-a_2}, \quad R^{-1}(z)=I-\frac{R_0}{z-a_1},
\quad \d(R(z))=\frac{z-a_1}{z-a_2}, $$
$$R_0=\frac{(a_2-a_1)w\cdot w_w^T}{\langle w, w\textprime\rangle}.$$

The choice of $w$ and $w\textprime$ for all singular points of $A(z)$
determines bases in the corresponding spaces for all
deformations. More precisely, if $\hat {\hat A}(z)$ is the next
deformation such that $$(\hat a_1, \hat a_2, \dots, a_n)\rightarrow
(q\hat a_1, q\hat a_2, \dots, a_n),$$
$$\hat w=R(qa_1)A^{-1}(qa_1)w, \quad \hat w\textprime=w\textprime^T
A(qa_2)R^{-1}(qa_2).$$

Now the second ratio of the $\tau$- function is equal to 
$$\frac{\langle \hat w, \hat w\textprime\rangle}{\langle w, w\textprime\rangle}$$
and does not depend on the choices we made.

\begin{thm}\label{tau}
 Consider $q^{-1}$-connections on $\L=\O(-1)\oplus\O(-1)$ of
type
$$\lambda=(a_1, a_3, a_5;a_2, a_4, a_6; u, q^{-1}v; w, w;3).$$
Then their matrices can be written in the following form
\begin{equation} \label{A}
A=\frac{1}{(z-a_2)(z-a_4)(z-a_6)}\left [\begin{array}{cc} uz^3+g_2z^2+g_1z+w & k_2z^2+k_1z\\ l_2z^2+l_1z & vz^3+h_2z^2+h_1z+w
\end{array}\right].
\end{equation}
Also, the following holds
$$\d (A(z))=uv\frac{(z-a_1)(z-a_3)(z-a_5)}{(z-a_2)(z-a_4)(z-a_6)}.$$

We consider such matrices modulo $q$-gauge transformations of the form
$\eqref{gauge},$ where 
$$R(z)=\left [\begin{array}{cc} r_{11} & r_{12}\\ 0 & r_{22}
\end{array}\right];$$
$$r_{11}=const,\quad deg(r_{12})\leq 1, \quad r_{22}=const.$$
Consider the second root of the matrix element $a_{21}$ $($note that $0$ is
always a root of $a_{21}$$)$ as the first
coordinate on $M_{\lambda}$, denote it $t$. Let the second coordinate
be $$r=\frac{w(t-a_3)(t-a_4)(t-a_5)(t-a_6)}{a_{11}(t)a_3a_4a_5a_6t}-\frac{1}{t}$$
Consider the modification of $\L$ to
$\hat\L$ that shifts 
$$a_1\rightarrow \hat a_1=qa_1, \quad a_2\rightarrow\hat
a_2=qa_2,\quad w\rightarrow \hat w=qw.
$$
Then the coordinates $(\hat t,\hat r)$ on the moduli space $M_{\hat\lambda}$ are related
to $(t,r)$ by $(\ref{recc1})$. Moreover, if we consider the modification
of $\hat \L$ that shifts 
$$\hat a_1\rightarrow q\hat a_1, \quad \hat a_2\rightarrow q\hat
a_2,\quad \hat w\rightarrow q \hat w
$$ the ``second logarithmic derivative'' of the
$\tau$-function makes sense as a number and the following holds
\begin{equation}
\D^2\tau=\frac{\tau(\L)\tau(\hat{\hat \L})}{(\tau{\hat
    (\L)})^2}=\frac{(q\hat r w-va_1a_2)(qrw-ua_1a_2)(q\hat t-a_1)(\hat
  t-a_2)}{(a_1-qa_3)(a_1-qa_5)(a_2-qa_4)(a_2-qa_6)a_1a_2}.
\end{equation}
\end{thm}

\begin{proof}
The first part of the theorem follows from Theorem \ref{qPV}. Note that
there is a natural modification $\O(-1)\oplus O(-1)\rightarrow
\O\oplus \O(-1)$ such that the corresponding parameters of the
$q$-connections change accordingly. 

Next, to compute the ``second logarithmic derivative'' of the
$\tau$-function we just have to follow the algorithm described above.   
\end{proof}

\section{Gap probabilities and the $\tau$-function}
\label{sec:6}

\subsection{The q-Hahn ensemble.}

Let $q\in(0,1)$ and $N\in \mathbb Z_{>0}$. Let $0<\a<q^{-1}$ and
$0<\b<q^{-1}$ or $\a>q^{-N}$ and $\b>q^{-N}.$ Fix a positive integer
$k.$ Recall that the weight
function for the $q$-Hahn ensemble is given by
$$\omega(x)=(\a\b q)^{-x}\frac{(\a q,q^{-N};q)_x}{(q,\b^{-1}N^{-1};q )_x},$$ where $(y_1,
\dots, y_i;q)_k=(y_1;q)\cdots (y_i;q)_k$ and
$(y;q)_k=(1-y)(1-yq)\cdots (1-yq^{k-1})$ is the
$q$-Pochhammer symbol.

We consider a probability distribution on the set of all $k$-subsets (often called
particle configurations) of
$$\mathfrak X= \{y_{i}=q^{-i}: i=0,\dots, N\}$$ given by 
$$
\mathcal P(y_{x_1},\dots ,y_{x_k})=
\frac{1}{Z} \prod_{1\leq i<j \leq k} {(q^{-x_i}-q^{-x_j})^2}\cdot \prod_{i=1}^{k}{w(x_i)},
$$
where $Z$ is a normalizing constant equal to
$\sum\limits_{y_{x_1},\dots ,y_{x_k}\in \X}^{ }\mathcal
P(y_{x_1},\dots ,y_{x_k})$. 

\begin{defi}
Given a particle configuration $(y_{x_1},\dots ,y_{x_k})$ we will call
$y_{\max\{x_i\}}$ the rightmost particle and $y_{\min\{x_i\}}$ the
leftmost particle respectively.
\end{defi}

We are interested in the gap probabilities 
$$D_k(s)=\sum_{x_1,\dots, x_k< s}\mathcal P(y_{x_1}, \dots,
y_{x_k}).$$

We will need the following fact from the general theory of the
orthogonal polynomial ensembles.
Consider a
kernel $\mathcal K$ on $\X\times\X$ defined by

$$
\mathcal K(y_i,y_j) = \begin{dcases}
        \sqrt {\omega(i)}\sqrt{\omega(j)} \frac{u(y_i)v(y_j)-v(y_i)u(y_j)}{y_i-y_j},
            & \quad i\neq j, \\
\quad \\
             \omega(i)(\dot u(y_i)v(y_i)-u(y_i)\dot v(y_i)),&\quad i=j.
      \end{dcases} 
$$
Here $u(\zeta)=P_k(\zeta)$ and
$v(\zeta)=(P_{k-1},P_{k-1})^{-1}_{\w}\cdot P_{k-1}(\zeta)$, where $P_k$
is the monic orthogonal polynomial of degree $k$ associated to the
weight function and $$(f(\zeta)g(\zeta))_w\colon=\sum_{y_x\in
  \X}f(y_x)g(y_x)\w(x)$$ is the corresponding inner product on the space of
all polynomials. 

Let us denote $\{y_s,y_{s+1},\dots, y_{N}\}$ by $\N_s$.
It is well known that the gap probability can be expressed as a Fredholm
determinant
$$D_k(s)=\d(1-\mathcal K_s), \quad s\in
\mathbb Z_{\geq k},\quad s\leq N,$$
where $\mathcal K_s$ is the restriction of $\mathcal K$ to $\N_s\times\N_s.$

\subsection{Geometric construction. }
We will closely follow the exposition in \cite{BA2}(Section
5) adapting it for our specific needs.
Let us consider on $\P^1$ a vector bundle
$$\mathcal L_{\varnothing}=\mathcal O(k-1)\oplus \mathcal O(-k-1). $$
For any subset $\mathfrak M\subset \ X$, define the following modification:
\begin{itemize}
\item $\mathcal L_{\mathfrak M}$ and $\mathcal L$ coincide on $\mathbb P^1 \setminus \mathfrak M$;
\item Near any $y\in \mathfrak M$ sections of $\mathcal L_{\mathfrak M}$ are rational sections $s=(s_1,s_2)^t\in \mathcal L$, such that $s_2$ is regular at $y$ and $s_1$ has at most a first order pole at $y$ with $\res_y(s_1)=\omega(y)s_2(y)$.
\end{itemize}

\begin{lem}[\cite{BA2}, Proposition 5.2]
Under the above assumptions $\L_{\X}\simeq \O(-1)\oplus \O(-1).$
\end{lem}

Let $\L^{up}_{\mathfrak M}$ be a modification of $\L_{\mathfrak M}$ on
$\X-\mathfrak M$ whose sections near $x\in \X-\mathfrak M$ are of the
form 
$$s=(s_1,s_2)\in \L_{\mathfrak M},$$
where $s_1$ is regular at $x$, $s_2$ has at most a first order pole
at $x$ and $\res_x s\sim w_x=\begin{pmatrix}1\\0\end{pmatrix}.$ Note
that $\L^{up}_{\mathfrak M}$ is also an upper modification of $\L_{\X}$.

For every point $x\in \X-\mathfrak M$ consider two functionals on
sections of $\L^{up}_{\mathfrak M}$:
$$f_x(s)=(w\textprime_x+w\textprime\textprime_x(z-x))s|_{z=x},$$
$$g_x(s)=\res_x s\cdot w_x,$$
where $w\textprime=\begin{pmatrix}\w(x)\\0\end{pmatrix}$ and
$w\textprime\textprime=\begin{pmatrix}1\\0\end{pmatrix}$.
Note that the sections of $\L_{\X}$ are exactly sections on
$\L^{up}_{\mathfrak M}$ on which $f_x$ (resp. $g_x$) vanish for all
$x\in \X-\mathfrak M$.
In this way, assuming that $|\X|$ is finite, one gets the identifications
$$f_{\X-\mathfrak M}=(f_x)_{x\in\X-\mathfrak
  M}\colon\H^0(\P^1,\L^{up}_{\mathfrak M}/\L_{\X})\simeq \mathbb
C^{|\X|-|\mathfrak M|},$$
$$g_{\X-\mathfrak M}=(g_x)_{x\in\X-\mathfrak
  M}\colon\H^0(\P^1,\L^{up}_{\mathfrak M}/\L_{\mathfrak M})\simeq \mathbb
C^{|\X|-|\mathfrak M|}.$$
This induces an isomorphism
$$\dr(\L_{\X})\otimes \dr(\L_{\mathfrak
  M})^{-1}=\d(\H^0(\P^1,\L^{up}_{\mathfrak M}/\L_{\X}))\otimes \d(\H^0(\P^1,\L^{up}_{\mathfrak M}/\L_{\mathfrak M}))^{-1}=\mathbb
C.$$
Thus, one can interpret the ratio $\frac{\tau(\L_{\mathfrak M})}{\tau(\L_{\X})}$ as the determinant of the composition 
$$\mathbb
C^{|\X|-|\mathfrak M|}\simeq \H^0(\P^1,\L^{up}_{\mathfrak
    M}/\L_{\X})\simeq  \H^0(\P^1,\L^{up}_{\mathfrak M})\rightarrow
  \H^0(\P^1,\L^{up}_{\mathfrak M}/\L_{\mathfrak M})\simeq \mathbb
C^{|\X|-|\mathfrak M|}.$$
 
\begin{thm}[\cite{BA2}, Theorem 5.3] \label{Fred} For any $\mathfrak M\subset \X$,
  the following holds 
\begin{equation}\label{ker}
\frac{\tau(\L_{\mathfrak M})}{\tau(\L_{\X})}=\d(1-\mathcal
K|_{\ell^{2}(\X-\mathfrak M)}).
\end{equation}
\end{thm}

\begin{prop}
Consider $\mathcal L_{\varnothing}=\mathcal O(k-1)\oplus \mathcal
O(-k-1)$ and $\mathfrak M_s=\{1, q^{-1},\dots, q^{-s}\}\subset \X$,
where $s\geq k.$ Then
$\L_{\mathfrak M}\simeq \mathcal O(-1)^2.$
\end{prop}
\begin{proof}
From the construction it follows that $\textup{deg}(\L_{\mathfrak M})=-2$,
thus, it is enough to show that $\L_{\mathfrak M}$ has no global
sections.

The section of $\L_{\mathfrak M}$ is of the form $s=(s_1, s_2),$ such
that $s_1$ is of degree at most $s-1$ and $s_2$ is given by 

\begin{equation}
s_2=\sum_{i=0}^{s-1}\frac{w(q^{-i})s_1(q^{-i})}{z-q^{-i}},
\end{equation}
 and its order of zero at $z=\infty$ is at least $k+1$.

 Equivalently, it means that for any polynomial $p(z)$ of degree at
 most $k-1$,
$$\lim_{z\rightarrow
  \infty}p(z)s_2(z)=\sum_{i=0}^{s-1}p(q^{-i})w(q^{-i})s_1(q^{-i})=0,$$
which is impossible.
\end{proof}

\subsection{$q^{-1}$-Connection associated to $q$-Hahn ensemble.}

Consider the following $q^{-1}$-connection on $\mathcal L_{\varnothing}$
$$ A_0(z)=\left [\begin{array}{cc}\frac{q\omega(q^{-1} z)}{\omega(z)}&0 \\0&1 \end{array}\right]  =
\left [\begin{array}{cc} \frac{(z-\a q)\cdot(z-q^{-M})}{\a\b (z-q)\cdot (z-b^{-1}q^{-M})}&0\\0&1 \end{array}\right].$$ 
Then the corresponding $q^{-1}$-connection on $\L_{\mathfrak M_s}$
is given
by 

$$A_s(z)=\mathcal M_s(q^{-1}z)A_0(z)\mathcal M_s^{-1}(z),$$
where $\mathcal M_k(z)$ is the unique solution
of the following normalized Riemann-Hilbert problem:
\begin{itemize}
\item It is analytic in
$\mathbb C - \{1, \dots q^{-k}\};$
\item $\underset{\text{  } z= q^{-i}}{\res} \mathcal M_k (z)=\underset{z\rightarrow q^{-i}}{\lim}
  \mathcal M_k\left
    [\begin{array}{cc} 0& \w(i) \\0&0 \end{array}\right ];$
\item $\mathcal M_k(z) \left [\begin{array}{cc} z^{-k}&0
      \\0& z^k\end{array}\right ]=I + O( \frac{1}{z} ), \text{ as }
  z\rightarrow \infty.$
\end{itemize}

In \cite{BB} it was shown that
$$\mathcal M_s(z)=\left [\begin{array}{cc} P_k(z)
    &\sum\limits_{u\in\{1,\dots, q^{-s+1}\}}^{ }
    \frac{P_k(u)\omega(u)}{z-u}\\cP_{k-1}(z)&\sum\limits^{ }_{u\in\{1,\dots, q^{-s+1}\}} \frac{P_{k-1}(u)\omega(u)}{z-u} \end{array}\right ],$$
where $P_k$ is the monic orthogonal polynomial of degree $k$ with respect to the $q$-Hahn
weight on $\{ 1, \dots, q^{-s+1}\}$ and
$c=(P_{k-1},P_{k-1})^{-1}_{\omega}.$

Let us compute $A_k(z).$ 
By substitution one can check that

$$ M_k(z)=\left [\begin{array}{cc} \prod\limits_{i=0}^{k-1}(z-q^{-i})&0
    \\ \prod\limits_{i=0}^{k-1}(z-q^{-i})
    \sum\limits_{i=0}^{k-1}\frac{\rho_i}{z-q^{-i}}&\prod\limits_{i=0}^{k-1}(z-q^{-i})^{-1} \end{array}\right],$$
where $\rho_i=\frac{1}{\w(i)}\prod\limits_{0\leq j\leq k-1}^{
}\frac{1}{q^{-i}-q^{-j}}.$

Thus,
\begin{multline*}
 A_k(z)=
\frac{1}{(z-q^{-k+1})(z-\b^{-1}q^{-N})(z-q)}\\
\left[\begin{array}{cc} \frac{q^{-k}}{\a\b}(z-q)(z-q^{-N})(z-\a q) & 0\\
   a^k_{21}(z) &
   q^k(z-q^{-k+1})^2(z-\b^{-1}q^{-N}) \end{array}\right], 
\end{multline*}
where 
\begin{equation} a^{k}_{21}=\frac{q^{-k}}{\a\b}\cdot (z-\a
  q)(z-q^{-N})(z-q)\sum\limits_{i=0}^{k-1}\frac{\rho_i}{q^{-1}z-q^{-m}}-q^k(z-q^{-k+1})^2(z-\b^{-1}q^{-N})\sum\limits_{i=0}^{k-1}\frac{\rho_i}{z-q^{-m}}. \nonumber
\end{equation}

Let us denote
$a_1=q^{-k+1};$ $a_2=q^{-k+1};$ $a_3=q^{-N};$ $a_4=\b^{-1}q^{-N};$
$a_5=\a q;$
$a_6=q.$

Summarizing, we get the following statement.

\begin{prop}\label{initial}
The matrix of the $q^{-1}$-connection $A_k(z)$ is of the following form 
\begin{equation} 
A_k(z)=\frac{1}{(z-a_2)(z-a_4)(z-a_6)}\left [\begin{array}{cc}
    \frac{q^{-k}}{\a\b}z^3+g_2z^2+g_1z+w & 0\\ f_2z^2+f_1z &
    q^kz^3+h_2z^2+h_1z+w \end{array}\right] \nonumber
\end{equation}
with
$$\d(A_k(z))=\frac{(z-a_1)(z-a_3)
(z-a_5)}{\a\b(z-a_2)(z-a_4)(z-a_6)}.$$
Introducing the usual coordinates $t$ and $r$, we can directly compute
\begin{itemize}
\item $w=\b^{-1}q^{-N-k+2}$;
\item $r_k=-a_4^{-1};$
\item $t_k=-\frac{f_1}{f_2}$, where
\begin{align*} 
 f_1=\sum\limits_{m=0}^{k-1}\rho_m&(\frac{q^{-k+1}}{\a\b}(q^{2(-m+1)}-(a_3+a_5+a_6)q^{-m+1}+a_3a_5+a_3a_6+a_5a_3)\\
& \text{          }-q^k(q^{-2m}-(a_1+a_2+a_4)q^{-m}+a_1a_2+a_1a_4+a_2a_4));\\
f_2 =\sum\limits_{m=0}^{k-1}\rho_m&(\frac{q^{-k+1}}{\a\b}(q^{-m+1}-a_3-a_5-a_6)-q^k(q^{-m}-a_1-a_2-a_4)). 
\end{align*}
\end{itemize}
\end{prop}

Notice that the isomonodromy transformation $\A(s) \rightarrow
\A(s+1)$ shifts $$a_1\rightarrow \hat a_1=qa_1, \quad a_2\rightarrow\hat
a_2=qa_2,\quad w\rightarrow \hat w=qw.$$

The only thing left in order to apply Theorem \ref{tau} to the series of
modifications $A_s(z)$ and prove Theorem \ref{qH} is to compute the initial conditions, i.e., $D_{k}(k),$ $D_{k}(k+1).$

\begin{prop}\label{D} $($\cite{BB}, Proposition 6.6$)$  Let $\X \subset \mathbb R$
  be a discrete set, let $\{P_n(z)\}$ be the family of orthogonal
  polynomials corresponding to a strictly positive weight function
  $\omega : \X \rightarrow \mathbb R$. Then
\begin{equation}
D_k(k) = \frac{1}{Z} \cdot \prod\limits_{0\leq i\le j\leq k−1}^{ } (\pi_i - \pi_j )^2\cdot \prod\limits_{l=0}^{k-1} \omega(l),
\end{equation}

\begin{equation}
D_k(k+1) = \omega(k) \cdot q_{k}^{-1}\cdot D_k(k)\cdot
\prod\limits_{l=0}^{k-1}(\pi_k - \pi_l)^2,\end{equation} 
where $q_k$ is given by 
$$q_k=(\rho_k+\sum\limits_{m=0}^{k-1}\frac{\rho_m}{(\pi_k-\pi_m)^2}).$$
\end{prop}

\begin{rem}
Using the standard arguments one can show that $Z$ is equal to
$\prod\limits_{i=0}^{k-1}(P_i,P_i)_\omega.$
\end{rem}

We are finally ready to prove Theorem \ref{qH}.
\begin{proof}
What we need to show is that the second logarithmic derivative of the
$\tau$-function for the series of isomonodromy transformations
$A_s(z)$ coincides with the second logarithmic derivative of $D_k(s).$
If so, the theorem will follow from Theorem \ref{tau}.

The $q$-connection $A_0(z)$ was chosen in such a way that we could
apply Theorem
\ref{Fred} to obtain the result, see \cite{BA2}.
\end{proof}

\section{Numerical computations}
\label{sec:7}
The goal of this section is to show the numerical evidence supporting
the conjecture, stating that the distribution of the rightmost
(resp. leftmost) particle in the $q$-Hahn ensemble converges to the
Tracy-Widom distribution after proper centering and scaling. In the
corresponding tiling model this describes fluctuations around the
frozen boundary. We do not
have the rigorous proof of this conjecture, but we expect that it can
be proved using the same approach as the one used in \cite{Baik} to
prove the corresponding statement for the Hahn ensemble. 

Let us fix $a, b<0$ and $0<k_0, q_0< 1.$ Assume that the
conjecture is true, namely, that there exist such constants $c_1$ and
$c_2$ that $D_{k}^{\alpha,\beta, q, N-1}(c_1N+c_2N^{1/3}u)$ converges
to $F_2(u)$ as $N\rightarrow \infty,$ where $F_2$ is the Tracy-Widom
ditribution, $q=q_0^{\frac{1}{N}},$ $k=k_0N,$
$\alpha=q^{\frac{a}{N}},$ and $\beta=q^{\frac{b}{N}}.$ We write
$D_{k}^{\alpha,\beta, q, N-1}$ to empathize the dependence on the
parameters $\alpha,$ $\beta,$ $q,$ $N.$  Below we present a heuristic computation of the constants $c_1$ and $c_2,$ using a method suggested by Olshanski in \cite{O}. 

Recall that the $q$-Hahn ensemble is the determinantal process with
correlation kernel $K(x,y)$ expressed through the monic orthogonal polynomials $P_0=1,P_1, \dots$:
$$K(x,y)=\sum\limits_{i=0}^{k-1}\tilde {P_i}(x) \tilde{P_i}(y), \quad x,y\in \mathfrak X=\{q^{-j}| j=0,\dots, N-1\},$$ 
where
$$\tilde{P_i}(x)=\sqrt{\omega(x)}P_i(x).$$

We will construct a selfadjoint operator $\mathfrak D$ from the difference operator for the $q$-Hahn ensemble in such a way that:
\begin{enumerate}[(i)]
\item $\mathfrak D$ acts on $L^2(\mathfrak X);$
\item the correlation kernel $K$ can be realized as a spectral projection $\mathcal P(\Delta),$ here $\Delta$ is a certain part of the spectrum of the operator $\mathfrak D;$
\item $\mathfrak{D}$ formally converges to the Airy operator $\mathfrak D^{\text{Airy}}g(u)=\ddot{g}(u)-ug(u)$ as $N\rightarrow \infty.$
\end{enumerate}

In this approach the next step is to prove a stronger convergence of the operators and then deduce the convergence of the correlation kernels, see \cite{BG}, \cite{G}. Unfortunately, we do not know how to do it rigorously at the moment. Nevertheless, in this way we can guess the scaling constants and check if they agree with numerical computations.  

Recall that the following $q$-difference equation holds \cite{KS}
$$q^{-n}(1-q^n)(1-\alpha\beta q^{n+1})P_n(x)=B(x)P_n(x+1)-(B(x)+D(x))P_n(x)+D(x)P_n(x-1),$$
where
\begin{align*}
B(x)&=(1-q^{x-N+1})(1-\alpha q^{x+1}),\\
D(x)&=\a\b q(1-q^x)(1-\b^{-1}q^{x-N}).
\end{align*}
The corresponding $q$-difference equation for $\tilde{P}_n(x)$ is
\begin{multline*}
q^{-n}(1-q^n)(1-\a\b q^{n+1})\tilde{P}_n(x)=B(x)\sqrt{\frac{\omega(x)}{\omega(x+1)}}\tilde{P}_n(x+1)\\
-(B(x)+D(x))\tilde{P}_n(x)+D(x)\sqrt{\frac{\omega(x)}{\omega(x-1)}}\tilde{P}_n(x-1).
\end{multline*}

Let us denote by $\tilde{\mathfrak D}$ the operator given by the right hand side. We are interested in the projection on the eigenvectors corresponding to $n\in\{0,\dots, N-1\}.$ Thus, the desired projection $\mathcal P$ is the spectral projection associated with the following segment:

$$[ q^{-(N-1)}(1-q^{N-1})(1- \a\b q^{N}), 0 ].$$

Since eventually we want to get the Airy kernel, which corresponds to the spectral projection onto the positive part of the spectrum of the Airy operator, instead of $\tilde {\mathfrak D}$ we should consider
$$\mathfrak D=\tilde {\mathfrak D}- q^{-(N-1)}(1-q^{N-1})(1- \a\b q^{N}) \mathfrak I,$$
where $\mathfrak I$ is the identity operator.

Let $g(u)$ be a smooth function on $\mathbb R.$ Assign to it a function $f(x)$ on $\mathfrak X$ by setting $f(x)=g(u),$ where $x$ and $u$ are related by $x=c_1 N+c_2 N^{1/3}u.$

Then we get 
$$ f(x\pm 1)=g(u\pm c_{2}^{-1} N^{-\frac{1}{3}})\approx g(u)\pm c_{2}^{-1} N^{-\frac{1}{3}}+\frac{1}{2}c_2^{-2} N^{-\frac{2}{3}}. $$

We can find two values of $c_1$ (these corresponds to the rightmost and the
leftmost particle) and $c_2$ such that $\mathfrak D\approx c
N^{-\frac{2}{3}}(\ddot{g}(u)-ug(u)),$ where $c$ is a constant. More
precisely, $$c_1=\frac{\log(p)}{\log(q_0)}$$
and $p$ is equal to
\begingroup\makeatletter\def\f@size{9}\check@mathfonts
$$\frac{(A^2B^2q_0+A^2Bq_o+A^2B+AB)K^3+(-2A^2Bq_0-2ABq_0-2AB-2A)K^2+(ABq_0+Aq_0+A+1)K\pm2q_0\sqrt{D}}{K(4A^2Bq_0K^2 +(A^2B^2p^2-2A^2Bq_0^2-2A^2Bq_0+A^2q_0^2-2A^2q_0-2ABq_0+A^2-2Aq_0-2A+1)K +4Aq_0)},$$
\endgroup
where 
\begin{multline*}
D=K(K^2-Kq_0-K+q_0)(1-BK)A\\
\cdot(A^3B^2K^3q_0-A^2B^2K^2q_0-A^2BK^2q_0-A^2BK^2+ABKq_0+ABK+AK-1)
\end{multline*}
and $A=q_0^a, \text{ } B=q_0^b, \text{ } K=q_0^{k_0}.$
The constant $c_2$ can be expressed through a unique real solution of
a cubic equation. We will not present here the formula for $c_2$ since it is quite
lengthy. 

In the limit the eigenvalue equation $\mathfrak D\psi=\lambda \psi$ 
turns into the equation
$\mathfrak D^{\text{Airy}}=s\psi$
 after the renormalization $\lambda=c N^{-\frac{2}{3}}s.$
 
Next, we present some results of our computations. First, we plot the
density function
$$(D_k(c_1N+c_2uN^{\frac{1}{3}}+1)-D_k(c_1N+c_2uN^{\frac{1}{3}}))\cdot
c_2N^\frac{1}{3}$$
for the same values of the parameters but different values of $N.$ The
blue graph is the density function of the Tracy-Widom distribution.\footnote{To
  plot the density of the Tracy-Widom distribution we used a table of
  its values
 \url{ http://www.wisdom.weizmann.ac.il/~nadler/Wishart_Ratio_Trace/TW_ratio.html}
based on \cite{Bo}.}

 \begin{figure}[H]
\includegraphics[width=0.33\linewidth]{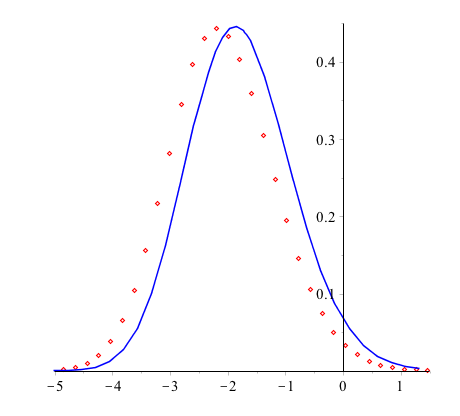}
\caption{Plot of the density function of the distribution of the
  rightmost particle with parameters $N=2000,$ $q=0.99,$ $k_0=0.2,$ $a=-1.1,$ $b=-1.3,$ $c_1=0.84839,$ $c_2=0.38999.$}
  \label{fig:501}
\end{figure}  

\begin{figure}[H]
\includegraphics[width=0.40\linewidth]{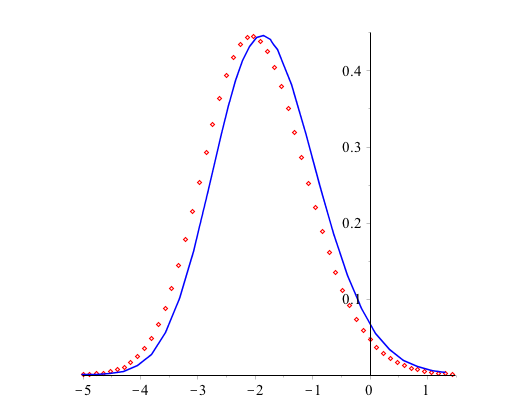}
\caption{Plot of the density function of the distribution of the
  rightmost particle with parameters $N=10000,$ $q=0.99,$ $k_0=0.2,$
  $a=-1.1,$ $b=-1.3,$ $c_1=0.84839,$ $c_2=0.38999.$}
  \label{fig:501}
\end{figure} 

We also plot $\mathbb E(\textup{TW})-\mathbb E(\textup{qH}_N)$  in logarithmic coordinates, where $\textup{TW}$ is the Tracy-Widom distribution
and $\textup{qH}_N$ is the distribution of the rightmost particle for
the $q$-Hahn ensemble. The following
pictures illustrate convergence of this graph to a line with 
slope $-\frac{1}{3}$, which agrees with our expectations \cite{F}. 
\begin{figure}[H]
\includegraphics[width=0.40\linewidth]{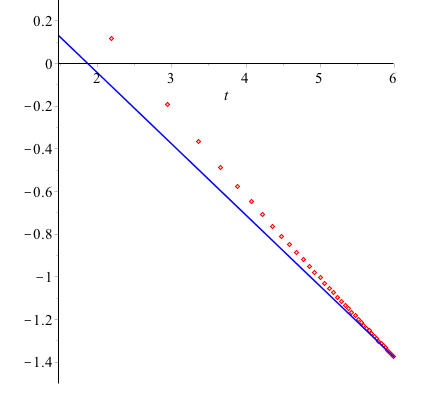}
\caption{$N=10, 20,
\dots, 400;$ $q_0=0.99,$ $k_0=0.3,$
$a=-1.1,$ $b=-1.3,$ $c_1=0.84839,$ $c_2=0.38999.$  The blue line has slope $-\frac{1}{3}.$}
  \label{fig:990}
\end{figure}  

\begin{figure}[H]
\includegraphics[width=0.40\linewidth]{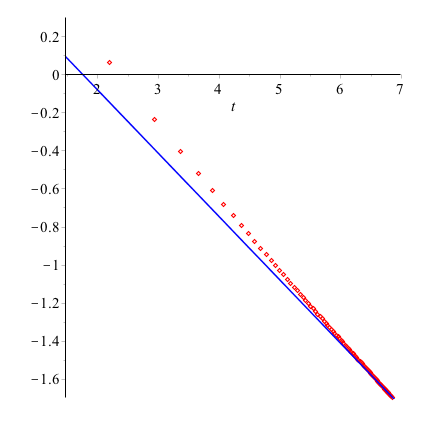}
\caption{$N=10, 20,
\dots, 1000;$ $q_0=0.5,$ $k_0=0.2,$
$a=-1.1,$ $b=-1.3,$ $c_1=0.69758,$ $c_2=0.47101.$ }
  \label{fig:501}
\end{figure}

\end{document}